\documentclass[12pt,letterpaper]{article}
\usepackage{amsmath,amsfonts,amsthm,amssymb,stmaryrd,relsize}

\theoremstyle{plain}

\newtheorem{thm}{Theorem}
\newtheorem{lem}[thm]{Lemma}
\newtheorem{cor}[thm]{Corollary}

\newenvironment{exam}[1]
{\begin{flushleft}\textbf{Example #1}.\enspace}%
{\end{flushleft}}

\allowdisplaybreaks  

\newcommand{\complex}{{\mathbb C}}
\newcommand{\real}{{\mathbb R}}

\newcommand{\rmtr}{\mathrm{tr\,}}

\newcommand{\ab}[1]{\left|#1\right|}
\newcommand{\doubleab}[1]{\left|\left|#1\right|\right|}
\newcommand{\brac}[1]{\left\{#1\right\}}
\newcommand{\paren}[1]{\left(#1\right)}
\newcommand{\sqbrac}[1]{\left[#1\right]}
\newcommand{\elbows}[1]{{\left\langle#1\right\rangle}}
\newcommand{\ket}[1]{{\left|#1\right>}}
\newcommand{\bra}[1]{{\left<#1\right|}}

\begin{document}

\title{A CHARACTERIZATION FOR\\
ENTANGLED VECTORS}
\author{Stan Gudder\\ Department of Mathematics\\
University of Denver\\ Denver, Colorado 80208\\
sgudder@du.edu}
\date{}
\maketitle

\begin{abstract}
This article presents a simple characterization for entangled vectors in a finite dimensional Hilbert space $H$. The characterization is in terms of the coefficients of an expansion of the vector relative to an orthonormal basis for
$H$. This simple necessary and sufficient condition contains a restriction and we also present a more complicated condition that is completely general. Although these characterizations apply to bipartite systems, we show that generalizations to multipartite systems are also valid.
\end{abstract}

Let $H_1$ and $H_2$ be finite-dimensional complex Hilbert spaces with $\dim H_1=m$, $\dim H_2=n$ and let $H=H_1\otimes H_2$. A vector $\psi\in H$ is \textit{factorized} if there are vectors $\alpha\in H_1$,
$\beta\in H_2$ with $\psi =\alpha\otimes\beta$ \cite{hz12,nc00}. If $\psi$ is not factorized, then $\psi$ is \textit{entangled} \cite{bus03,hhhh09,zb06}. When $\psi =\alpha\otimes\beta$ we call $\alpha$ and $\beta$ the \textit{local parts} of $\psi$. Our first result shows that the local parts are unique to within complex multiples. If the vectors are states (unit vectors) then the multiples have absolute have absolute value one (phase factors).

\begin{lem}    
\label{lem1}
If $\alpha ,\beta\ne 0$, then $\alpha\otimes\beta =\alpha '\otimes\beta '$ if and only if there exists a nonzero
$a\in\complex$ such that $\alpha '=a\alpha$ and $\beta '=\tfrac{1}{a}\,\beta$.
\end{lem}
\begin{proof}
Clearly, if $\alpha '=a\alpha$ and $\beta '=\tfrac{1}{a}\,\beta$, then $\alpha\otimes\beta=\alpha '\otimes\beta '$. Conversely, suppose that $\alpha\otimes\beta =\alpha '\otimes\beta '$. Then employing Dirac notation, we have
\begin{align*}
\ket{\alpha}\bra{\alpha}\otimes\ket{\beta}\bra{\beta}&=\ket{\alpha\otimes\beta}\bra{\alpha\otimes\beta}
   =\ket{\alpha '\otimes\beta '}\bra{\alpha '\otimes\beta '}\\
   &=\ket{\alpha '}\bra{\alpha '}\otimes\ket{\beta '}\bra{\beta '}
\end{align*}
Taking the partial trace over $H_2$ \cite{hz12,nc00} gives
\begin{align*}
\doubleab{\beta}^2\paren{\ket{\alpha}\bra{\alpha}}
   &=\paren{\ket{\alpha}\bra{\alpha}}\rmtr\paren{\ket{\beta}\bra{\beta}}
   =\paren{\ket{\alpha '}\bra{\alpha '}}\rmtr\paren{\ket{\beta'}\bra{\beta'}}\\
   &=\doubleab{\beta '}^2\paren{\ket{\alpha '}\bra{\alpha '}}
\end{align*}
Acting on $\alpha$ gives
\begin{equation*}
\doubleab{\alpha}^2\doubleab{\beta}^2\alpha =\doubleab{\beta '}^2\elbows{\alpha ',\alpha}\alpha '
\end{equation*}
We conclude that $\alpha '=a\alpha$ with
\begin{equation*}
a=\frac{\doubleab{\alpha}^2\doubleab{\beta}^2}{\doubleab{\beta '}^2\elbows{\alpha ',\alpha}}
\end{equation*}
In a similar way, there is a $b\in\complex$ with $\beta '=b\beta$. Hence,
\begin{equation*}
\alpha\otimes\beta =\alpha '\otimes\beta '= ab\alpha\otimes\beta
\end{equation*}
so that $b=1/a$.
\end{proof}

Let $\brac{\phi _1,\ldots ,\phi _m}$, $\brac{\psi _1,\ldots ,\psi _n}$ be orthonormal bases for $H_1$ and $H_2$, respectively, and let $\psi=\sum c_{ij}\phi _i\otimes\psi _j\in H_1\otimes H_2$. Can we tell, by examining the coefficients $c_{ij}$, whether $\psi$ is entangled or not?

\begin{exam}{1}  
Let
\begin{align*}
\psi &=4\phi _1\otimes\psi _1-3i\phi _1\otimes\psi _2+5\phi _1\otimes\psi _3-8\phi _2\otimes\psi _1
  +6i\phi _2\otimes\psi _2\\
  &\quad -10\phi _2\otimes\psi _3+12\phi _3\otimes\psi _1-9i\phi _3\otimes\psi _2+15\phi _3\otimes\psi _3
\end{align*}
Can we use the coefficients to determine if $\psi$ is factorized and if it is, do they give the local parts? The next theorem shows that the answer is yes.\hfill\qedsymbol
\end{exam}

\begin{thm}    
\label{thm2}
Let $\psi =\sum c_{ij}\phi _i\otimes\psi _j$ and suppose that $\sum c_{ij}\ne 0$. Then $\psi$ is factorized if and only if for all $i=1,\ldots ,m$, $j=1,\ldots ,n$ we have
\begin{equation}                
\label{eq1}
c_{ij}\sum _{i,j}c_{ij}=\sum _jc_{ij}\sum _ic_{ij}
\end{equation}
\end{thm}
\begin{proof}
We have that $\psi$ is factorized if and only if $\psi =\alpha\otimes\beta$ for some $\alpha\in H_1$,
$\beta\in H_2$. Let $\alpha =\sum a_i\phi _i$ and $\beta =\sum b_j\psi _j$. It follows that
\begin{equation*}
\sum _{i,j}c_{ij}\phi _i\otimes\psi _j=\paren{\sum a_i\phi _i}\otimes\paren{\sum b_j\psi _j}
  =\sum _{i,j}a_ib_j\phi _i\otimes\psi _j
\end{equation*}
Hence, $\psi$ is factorized if and only if there exist sequences of complex numbers $\brac{a_i}$, $\brac{b_j}$, $i=1,\ldots ,m$, $j=1,\ldots ,n$ such that $c_{ij}=a_ib_j$. If \eqref{eq1} holds, letting $c=\sum _{i,j}c_{ij}$,
$a_i=\tfrac{1}{c}\,\sum _jc_{ij}$, $b_j=\sum _ic_{ij}$ we have that $c_{ij}=a_ib_j$ so $\psi$ is factorized. Conversely, suppose that $\psi$ is factorized and hence there exist sequences $\brac{a_i}$, $\brac{b_j}$ with $c_{ij}=a_ib_j$. Then $\sum _jc_{ij}\sum _ic_{ij}=a_ib_j\sum _{i,j}a_ib_j=c_{ij}\sum _{i,j}c_{ij}$. Hence,
\begin{equation*}
\sum _jc_{ij}\sum _ic_{ij}=a_ib_j\sum _{i,j}a_ib_j=c_{ij}\sum _{i,j}c_{ij}
\end{equation*}
so \eqref{eq1} holds.
\end{proof}

Notice that when $\phi$ is factorized as $\phi =\alpha\otimes\beta$, then Theorem~\ref{thm2} gives the local parts
$\alpha =\sum a_i\phi _i$, $\beta =\sum b_j\psi _j$ where $a_i=\tfrac{1}{c}\,\sum _jc_{ij}$,
$b_j=\sum _ic_{ij}$. Of course, $\alpha$ and $\beta$ have the flexibility given by Lemma~\ref{lem1}. This can be used to normalize $\alpha$ and $\beta$ when $\psi$ is a vector state. Although Theorem~\ref{thm2} requires that $\sum c_{ij}\ne 0$, we do have the following result.

\begin{cor}    
\label{cor3}
For $\psi =\sum c_{ij}\phi _i\otimes\psi _j$, if $\sum c_{ij}=0$ and $\sum _jc_{ij}\sum _ic_{ij}\ne 0$, then $\psi$ is entangled.
\end{cor}

The only case not included in Theorem~\ref{thm2} and Corollary~\ref{cor3} is when
\begin{equation}                
\label{eq2}
\sum c_{ij}=\sum _jc_{ij}\sum _ic_{ij}=0
\end{equation}

We now show that when this happens, then no conclusion can be drawn.

\begin{exam}{2}  
Let $\psi =\phi _1\otimes\psi _1-\phi _2\otimes\psi _1-\phi _1\otimes\psi _2+\phi _2\otimes\psi _2$. Then \eqref{eq2} holds and $\psi$ is factorized because
\begin{equation*}
\psi =(\phi _1-\phi _2)\otimes (\psi _1-\psi _2)
\end{equation*}
On the other hand, let
\begin{equation*}
\psi '=\phi _1\otimes\psi _1-\phi _1\otimes\psi _2+\phi _2\otimes\psi _3-\phi _2\otimes\psi _4
\end{equation*}
Again, \eqref{eq2} holds, but $\psi '$ is entangled. To show this, replace $\psi _2$ by $-\psi _2$ to get
\begin{equation*}
\psi '=\phi _1\otimes\psi _1+\phi _1\otimes (-\psi _2)+\phi _2\otimes\psi _3-\phi _2\otimes\psi _4
\end{equation*}
Now $\sum c_{ij}=2\ne 0$, $c_{23}=1$, $\sum _jc_{2j}=0$, $\sum _ic_{i3}=1$ so \eqref{eq1} does not hold. It follows from Theorem~\ref{thm2} that $\psi '$ is entangled. Another way to see this is to write
\begin{equation*}
\psi '=\phi _1\otimes (\psi _1-\psi _2)+\phi _2\otimes (\psi _3-\psi _4)
\end{equation*}
We then have that $\sum c_{ij}=2\ne 0$, $c_{11}=1$, $\sum _jc_{1j}=\sum _ic_{i1}=1$ so \eqref{eq1} does not hold.\hfill\qedsymbol
\end{exam}

We now return to Example~1. We have that $c_{11}=4$, $c_{12}=-3i$, $c_{13}=5$, $c_{21}=-8$, $c_{22}=6i$, $c_{23}=-10$, $c_{31}=12$, $c_{32}=-9i$, $c_{33}=15$. We then obtain, $\sum _{i,j}c_{ij}=6(3-i)$,
$\sum _jc_{1j}=3(3-i)$, $\sum _jc_{2j}=6(i-3)$, $\sum _jc_{3j}=9(3-i)$, $\sum _ic_{i1}=8$,
$\sum _ic_{i2}=-6i$,
$\sum _ic_{i3}=10$. It is now easy to check that \eqref{eq1} holds for every $i,j$. Hence, by Theorem~\ref{thm2},
$\psi$ is factorized. Using our previous formulas, we have that, $a_1=1/2$, $a_2=-1$, $a_3=3/2$, $b_1=8$, $b_2=-6i$, $b_3=10$. Hence, $\psi =\alpha\otimes\beta$ where
\begin{equation*}
\alpha =\tfrac{1}{2}\,\phi _1-\phi _2+\tfrac{3}{2}\,\phi _3,\quad\beta =8\psi _1-6i\psi _2+10\psi _3
\end{equation*}
or slightly simpler, by Lemma~\ref{lem1}, we can let
\begin{equation*}
\alpha =\phi _1-2\phi _2+3\phi _3,\quad\beta =4\psi _1-3i\psi 2+5\psi _3
\end{equation*}

Recall that any $\psi\in H$, with $\psi\ne 0$, has a \textit{Schmidt decomposition}
\begin{equation}                
\label{eq3}
\psi =\sum _{i=1}^r\lambda _i\phi _i\otimes\psi _i
\end{equation}
where $\lambda _i>0$ and $\brac{\phi _i}$, $\brac{\psi _j}$ are orthonormal vectors in $H_1$ and $H_2$, respectively \cite{hz12,nc00}. It is well-known that $\psi$ is factorized if and only if $r=1$ \cite{hz12,nc00}.
If $r=1$, then clearly $\psi$ is factorized. The converse is not so obvious. We can employ Theorem~\ref{thm2} to prove this converse. The decomposition \eqref{eq3} gives an expansion with $c_{ij}=\lambda _i\delta _{ij}$. Then
$\sum _{i,j}c_{ij}=\sum\lambda _i\ne 0$, $\sum _jc_{ij}=\lambda _i$, $\sum _ic_{ij}=\lambda _j$. Then \eqref{eq1} becomes
\begin{equation}                
\label{eq4}
c_{ij}\sum _{i=1}^r\lambda _i=\lambda _i\lambda _j
\end{equation}
for all $i,j$. If $r=1$, then \eqref{eq4} becomes $\lambda _1^2=\lambda _1^2$ so $\psi$ is factorized. If
$r\ge 2$, letting $i=1$, $j=2$ we have for \eqref{eq4} that
\begin{equation*}
c_{12}\sum _{i=1}^r\lambda _i=\lambda _1\lambda _2
\end{equation*}
But this is impossible because $c_{12}=0$. Applying Theorem~\ref{thm2} shows that $\psi$ is entangled. The reader may wonder why we don't just use the Schmidt decomposition to test whether a general vector $\psi$ is factorized. One reason is that the Schmidt decomposition can be difficult to construct. Another reason is that our results generalize to multipartite systems where Schmidt decompositions are not available. This will be shown subsequently.

Although Theorem~\ref{thm2} requires that $\sum c_{ij}\ne 0$, this condition rarely holds and when it does, we can frequently alter one of the bases slightly so the condition does not hold. In fact, we already used this technique in Example~2. We now present a general result that is always valid. It is not as simple as Theorem~\ref{thm2} because now we have two conditions to verify. We denote the argument of complex number $c$ by $\arg (c)$ so that $c=\ab{c}e^{i\arg (c)}$. We assume that $0\le\arg (c)<2\pi$.

\begin{thm}    
\label{thm4}
If $\psi =\sum c_{ij}\phi _i\otimes\psi _j$, then $\psi$ is factorized if and only if for all $i,j$ we have
\begin{equation}                
\label{eq5}
\ab{c_{ij}}\sum _{i,j}\ab{c_{ij}}=\sum _j\ab{c_{ij}}\sum _i\ab{c_{ij}}
\end{equation}
and there exists a number $c\in\real$ such that for all $i,j$ we have
\begin{equation}                
\label{eq6}
\sum _j\arg (c_{ij})+\sum _i\arg (c_{ij})=\max (m,n)\arg (c_{ij})+c\pmod{2\pi}
\end{equation}
\end{thm}
\begin{proof}
We can assume without loss of generality that $m=n=d$ where $m=\dim H_1$, $n=\dim H_2$. Indeed, if $m<n$ say, then we can enlarge $H_1$ to a Hilbert space with dimension $n$. We can then consider $\psi$ of the form
$\psi =\sum _{i,j=1}^dc_{ij}\phi _i\otimes\psi _j$ where $c_{ij}=0$ if $i>m$. In this case, $\max (m,n)=d$. As in the proof of Theorem~\ref{thm2}, $\psi$ is factorized if and only if there exist sequences $\brac{a_i}$, $\brac{b_j}$, $i,j=1,\ldots ,d$ such that $c_{ij}=a_ib_j$. If $c_{ij}=a_ib_j$, then
\begin{equation*}
\sum _j\ab{c_{ij}}=\ab{a_i}\sum _j\ab{b_j},\quad\sum _i\ab{c_{ij}}=\ab{b_j}\sum _i\ab{a_i}
\end{equation*}
Hence,
\begin{equation*}
\sum _j\ab{c_{ij}}\sum _i\ab{c_{ij}}=\ab{a_i}\,\ab{b_j}\sum _i\ab{a_i}\sum _j\ab{b_j}=\ab{c_{ij}}\sum _{i,j}\ab{c_{ij}}
\end{equation*}
so \eqref{eq5} holds. We also have that
\begin{equation*}
\arg (c_{ij})=\arg (a_ib_j)=\arg (a_i)+\arg (b_j)\pmod{2\pi}
\end{equation*}
Hence, letting $c=\sum\arg (a_i)+\sum\arg (b_j)$ we have
\begin{align*}
\sum _j\arg (c_{ij})+\sum _i\arg (c_{ij})&=d\sqbrac{\arg (a_i)+\arg (b_i)}+c\\
   &=d\arg (c_{ij})+c\pmod{2\pi}
\end{align*}
Hence, \eqref{eq6} holds. Conversely, suppose that \eqref{eq5} and \eqref{eq6} hold. Letting
$a=\sum _{i,j}\ab{c_{ij}}>0$, $\ab{a_i}=\tfrac{1}{a}\,\sum _j\ab{c_{ij}}$, $\ab{b_j}=\sum _i\ab{c_{ij}}$ we obtain
$\ab{c_{ij}}=\ab{a_i}\,\ab{b_j}$. Moreover, letting 
\begin{equation*}
\alpha _i=\frac{1}{d}\,\sum _j\arg (c_{ij})-\frac{c}{d},\quad\beta _j=\frac{1}{d}\,\sum _i\arg (c_{ij})
\end{equation*}
we have that
\begin{align*}
\alpha _i+\beta _j&=\frac{1}{d}\,\sqbrac{\sum _j\arg (c_{ij})+\sum _i\arg (c_{ij})}-\frac{c}{d}
   =\arg (c_{ij})+\frac{c}{d}-\frac{c}{d}\\
   &=\arg (c_{ij})\pmod{2\pi}
\end{align*}
Defining $a_i=\ab{a_i}e^{i\alpha _i}$ and $b_j=\ab{b_j}e^{i\beta _j}$ we obtain
\begin{align*}
c_{ij}&=\ab{c_{ij}}e^{i\arg (c_{ij})}=\ab{a_i}\,\ab{b_j}e^{i(\alpha _i+\beta _j)}\\
   &=\ab{a_i}e^{i\alpha _i}\ab{b_j}e^{i\beta _j}=a_ib_j
\end{align*}
Hence, $\psi$ is factorized.
\end{proof}

Until now we have only considered bipartite systems. We now briefly discuss a multipartite system
$H=H_1\otimes H_2\otimes\cdots\otimes H_r$. A vecter $\psi\in H$ is \textit{factorized} if it has the form
$\psi =\alpha _1\otimes\alpha _2\otimes\cdots\otimes\alpha _r$, $\alpha _i\in H_i$, $i=1,\ldots ,r$. Otherwise,
$\psi$ is \textit{entangled}. Let $\phi _{j_i}^i$, $i=1,\ldots ,r$, $j_i=1,\ldots ,d_i$ be orthonormal bases for $H_i$. We then have that
\begin{equation}                
\label{eq7}
\psi =\sum c_{j_1j_2\ldots j_r}\phi _{j_1}^1\otimes\phi _{j_2}^2\otimes\cdots\otimes\phi _{j_r}^r
\end{equation}
We now generalize Theorem~2 to the multipartite case and leave the similar generalization of Theorem~\ref{thm4} to the reader.

\begin{thm}    
\label{thm5}
If $\psi$ has the form \eqref{eq7} and $\sum c_{j_1j_2\ldots j_r}\ne 0$, then $\psi$ is factorized if and only if for all $j_1,\ldots ,j_r$ we have
\begin{equation}                
\label{eq8}
c_{j_1\ldots ,j_r}\paren{\sum c_{j_1j_2\ldots j_r}}^{r-1}
   =\paren{\sum _{j_2,\ldots ,j_r}c_{j_1j_2\ldots j_r}}\cdots\paren{\sum _{j_1,\ldots ,j_{r-1}}c_{j_1j_2\ldots j_r}}
\end{equation}
\end{thm}
\begin{proof}
As in the proof of Theorem~\ref{thm2}, $\psi$ is factorized if and only if there exist sequences of complex numbers $\brac{a_{j_1}^1},\cdots ,\brac{a_{j_r}^r}$ such that
\begin{equation}                
\label{eq9}
c_{j_1j_2\ldots j_r}=a_{j_1}^1a_{j_2}^2\cdots a_{j_r}^r
\end{equation}
If \eqref{eq9} holds, then the right hand side of \eqref{eq8} becomes
\begin{align*}
&\paren{a_{j_1}^1\sum _{j_2,\ldots ,j_r}a_{j_2}^2\cdots a_{j_r}^r}
   \cdots\paren{a_{j_r}^r\sum _{j_1,\ldots ,j_{r-1}}a_{j_1}^1\cdots a_{j_r}^{r-1}}\\
   &\quad\quad(a_{j_1}^1\cdots a_{j_r}^r)\paren{\sum a_{j_1}^1\sum a_{j_2}^2\cdots\sum a_{j_r}^r}^{r-1}
   =c_{j_1\ldots j_r}\paren{\sum c_{j_1\ldots j_r}}^{r-1}
\end{align*}
so \eqref{eq8} holds. Conversely, suppose \eqref{eq8} holds. Letting
\begin{align*}
a_{j_1}^1&=\sum _{j_2,\ldots ,j_r}c_{j_1\ldots j_r}/\paren{\sum c_{j_1\ldots j_r}}^{r-1}\\
   \vdots&\\
   a_{j_r}^r&=\sum _{j_1,\ldots ,j_{r-1}}c_{j_1\ldots j_r}
\end{align*}
we have that \eqref{eq9} holds.
\end{proof}

\end{document}